\documentclass[letterpaper, 10 pt, conference]{ieeeconf}  

\IEEEoverridecommandlockouts                              

\usepackage[usenames,dvipsnames]{xcolor}
\usepackage{amsbsy,graphics,float}
\usepackage[dvips]{graphicx}
\usepackage{verbatim}

\usepackage{amsthm}
\usepackage{amsmath} 
\usepackage{amssymb}  
\usepackage{tikz}
\usetikzlibrary{arrows,positioning,shapes,fit,external} 
\usepackage{mathtools}

\usepackage{pgfplots}
\usepackage{pgfplotstable}
\usepackage{booktabs}
\pgfplotsset{compat = newest}
\usepackage{hyperref}

\title{\LARGE \bf
Gaussian Process based Passivation of a Class of \\Nonlinear Systems with Unknown Dynamics
}

\author{Thomas Beckers and Sandra Hirche
\thanks{T. Beckers and S. Hirche are with the Chair of Information-oriented Control (ITR), Department of Electrical and Computer Engineering,
Technical University of Munich, D-80333 Munich\newline
{\tt\small \{t.beckers, hirche\}@tum.de}}
}

\newtheorem{defn}{Definition}
\newtheorem{rem}{Remark}
\newtheorem{thm}{Theorem}
\newtheorem{lem}{Lemma}

\newtheorem{assum}{Assumption}

\newcommand\tran{\mkern-2mu\raise1.25ex\hbox{$\scriptscriptstyle\top\hspace{0.5mm}$}\mkern-3.5mu}
\newcommand{\R}{\mathbb{R}}
\newcommand{\N}{\mathbb{N}}

\newcommand{\X}{\mathcal{X}}

\newcommand{\bm}[1]{{\boldsymbol{#1}}}
\newcommand{\Verts}[1]{{\left\Vert #1 \right\Vert}}
\DeclareMathOperator{\diag}{diag}
\DeclareMathOperator{\var}{var}
\DeclareMathOperator{\mean}{\mu}
\DeclareMathOperator{\Var}{\Sigma}
\DeclareMathOperator{\Mean}{\bm\mu}

\newcommand{\dx}{\dot{\bm x}}

\newcommand{\x}{\bm x}

\newcommand{\iu}[1][]{\bm{u}_{#1}}

\newcommand{\dyn}{\bm{f}}
\newcommand{\y}[1][]{\bm{y}_{#1}}
\newcommand{\leig}{\underline{\lambda}}
\newcommand{\geig}{\bar{\lambda}}

\usepackage[noabbrev]{cleveref} 
\crefname{rem}{Remark}{Remarks}
\crefname{assum}{Assumption}{Assumptions}
\crefname{prop}{Proposition}{Propositions}
\crefname{cor}{Corollary}{Corollaries}
\crefname{lem}{Lemma}{Lemmas}
\crefname{thm}{Theorem}{Theorems}
\crefname{defn}{Definition}{Definitions}
\crefname{figure}{Fig.}{Fig.}
\Crefname{figure}{Figure}{Figures}
\crefname{equation}{}{}
\usepackage{autonum}
\begin{document}

\maketitle
\thispagestyle{empty}
\pagestyle{empty}

\begin{abstract}
The paper addresses the problem of passivation of a class of nonlinear systems where the dynamics are unknown. For this purpose, we use the highly flexible, data-driven Gaussian process regression for the identification of the unknown dynamics for feed-forward compensation. The closed loop system of the nonlinear system, the Gaussian process model and a feedback control law is guaranteed to be semi-passive with a specific probability. The predicted variance of the Gaussian process regression is used to bound the model error which additionally allows to specify the state space region where the closed-loop system behaves passive. Finally, the theoretical results are illustrated by a simulation.
\end{abstract}


\section{Introduction}
Passivity-based techniques allow the analysis and synthesis of large and complex systems because of the particular composition properties, e.g. the parallel and feedback interconnection of passive sub-systems gives a passive overall system. The passivity property is also helpful for the interconnection with other systems which are mostly unknown but assumed to be passive such as in telepresence systems~\cite{nuno2011passivity}, robot manipulation~\cite{erhart2015b} or physical human-robot interaction (pHRI)~\cite{de2008atlas}.
Hence, passive systems possess very useful and beneficial properties which make them so interesting in control theory and also in real-world applications. However, many modern engineering systems are not inherently passive or even stable, e.g. high-performance aircraft. Thus, to take advantage of the passivity properties, these systems need to be rendered passive by control, e.g. with suitable state feedback or using passivity-based control (PBC)~\cite{huang1999design}.\\
The arising problem is that these techniques require a suitable storage function or, at least, knowledge about the system dynamics and structure~\cite{ortega2013passivity,seron1994adaptive}.
However, the underlying dynamics are often hard to obtain using first-order principles because of the complexity of the system or the unacceptable time exposure of the modeling process. Especially in modern control applications such as autonomous robotics or human-centered control, the modeling process is very challenging or even unfeasible.\\
A promising approach to avoid these issues is provided by data-driven Gaussian process regression (GPR)~\cite{deisenroth2015gaussian}.
GPR is a supervised learning technique which combines several advantages. It requires only a minimum of prior knowledge for the regression of arbitrary complex functions since the complexity of the model scales with the amount of training data~\cite{rasmussen2006gaussian}. 
Additionally, it generalizes well even for small training data sets and it has a precise trade-off between fitting the data and smoothing.
In comparison to neural networks, GPR provides not only a mean function but also a predicted variance, and therefore a measure of the model fidelity based on the distance to the training data. This is a significant benefit since this information can be used for stability considerations~\cite{beckers:ifacwc2017}.\\

On the context of classical parametric dynamic system models, approaches for the passivation of linear and nonlinear systems are proposed in~\cite{fradkov2003passification,byrnes1991passivity}. However, both approaches assume that the underlying system dynamics and structure is known. The identification of dynamical systems with Gaussian processes is performed in~\cite{kocijan2005dynamic} but without considering stability or passivity. The stability of Gaussian process based systems is numerically evaluated in~\cite{vinogradska2016stability}. Recently, also analytical results about the stability are provided~\cite{beckers:ifacwc2017,Berkenkamp2016ROA}. However, all these approaches do not investigate the passivity of the closed loop system.\\
The contribution of the paper is the passivation of a class of nonlinear systems where the dynamics of the system is unknown. For this purpose, a Gaussian process regression is used to learn the unknown dynamics. The mean of the GPR is exploited for the feed-forward compensation of the dynamics. We show that the closed loop of the unknown dynamics, the GPR and a feedback control law is semi-passive with a specific probability. Additionally, we explicitly determine the state space region in which the systems behaves passive.\\
The remainder of the paper starts with Section~\ref{sec:pre} where the class of systems and GPR are introduced. Section~\ref{sec:modeling} describe the computation of the model error and the proof of semi-passivity. The method is validated in Section~\ref{sec:sim}.
\section{Preliminaries}
\label{sec:pre}
\subsection{Considered class of systems}
In this paper, we consider the class of nonlinear systems which are described by\footnotemark
\begin{align}
\dx&=\begin{bmatrix}
\dx_1\\ \dx_2
\end{bmatrix}
=\begin{bmatrix}
\x_2\\ \dyn(\x,\iu)
\end{bmatrix}\notag\\
\y[ex]&=c\x_1+\x_2,\,c\in\R_{>0}\label{for:system}
\end{align}
with the measurable state~$\x\in\X^2\subseteq\R^{2n}$ where~$\x_1,\x_2\in\X$ and the input~$\iu\in\R^n$ with~$n\in\N$. The continuous vector field~$\dyn\colon\X^2\times\R^n\to\R^n$ is assumed to be unknown. 
\begin{rem}
	This class of systems contains for example many electrical and mechanical systems which fulfill the Euler-Lagrange equation, e.g. robot manipulators. The systems do not need to be control affine. The output~$\y[ex]\in\R^n$ is often used in interconnection scenarios of mechanical systems where it represents a velocity plus scaled position feedback.	
\end{rem}
\footnotetext{\textbf{Notation:} Matrices are described with capital letters while vectors are denoted with bold characters. The term~$M_{:,i}$ denotes the i-th column of the matrix~$M$. The expression~$\mathcal{N}(\mu,\Sigma)$ is the normal distribution with mean~$\mu$ and covariance~$\Sigma$. The Euclidean norm is given by~$\Vert\cdot\Vert$ and the largest eigenvalue of a matrix by~$\geig$ and the smallest by~$\leig$.}
The problem is to find an input~$\iu$ such that the system~\cref{for:system} becomes passive.
\subsection{Semi-passivity}
The concept of passivity is well known whereas the theory of semi-passive system is less frequently used so that we recall the definition.
\begin{defn}
\label{def:semipassiv}
Following~\cite{pogromsky1998passivity}, the system~\cref{for:system} is called
	\begin{itemize}
		\item[1.] semi-passive in~$D_x$ if there exists a nonnegative function~$V\colon D_x\to\R_{\geq 0}$ where~$V(0)=0$ such that
		\begin{align}
			\dot{V}(\x)&=\frac{\partial V}{\partial \x_1}\x_2+\frac{\partial V}{\partial \x_2}\dyn(\x,\iu[ex])\notag\\
			&\leq\y[ex]^\top\iu[ex]-h(\x).
		\end{align}
		The passive output~$\y[ex]\in\R^n$ is state-dependent and the function~$h\colon D_x\to\R$ is nonnegative outside the ball~$B_r\subset D_x$ with radius~$r$, i.e.
		\begin{align}
			\exists r>0,\Verts{\x}\geq r\Rightarrow h(\x)\geq 0.
		\end{align}
		\item[2.] strictly semi-passiv in~$D_x$ if the system is semi-passive and the function~$h(\x)$ is positive outside some ball~$B_r$.
	\end{itemize}
\end{defn}
Hence, the behavior of semi-passive systems is comparable to passive systems outside the ball~$B_r$, see~\cref{fig:semipassiv}. 
Additionally, a feedback interconnection with another passive system has an ultimately bounded solution~\cite{pogromsky1998passivity} so that every trajectory of the closed-loop systems enters a compact set in finite time and remains there.
\subsection{Gaussian Process Regression}
Assume a vector-valued, nonlinear function~$\y=\bm f_{GP}(\x)$ with~$\bm f_{GP}\colon \R^n\to \R^n$ and~$\y\in\R^n$. The measurement values~$\tilde{\y}\in\R^n$ of the function are corrupted by Gaussian noise~$\bm\eta\in\R^n$, i.e.
\begin{align}
	\tilde{\y}&=\bm f_{GP}(\x)+\bm\eta\\
	\bm\eta&\sim\mathcal{N}(\bm 0,\diag (\sigma_{1}^2,\ldots,\sigma_{n}^2))
\end{align}
with the standard deviation~$\sigma_{1},\ldots,\sigma_{n}\in\R_{\geq 0}$. For the regression, the function is evaluated at~$m$ input values~$\{\bm x^{\{j\}}\}_{j=1}^m$. Together with the resulting measurements~$\{\tilde{\y}^{\{j\}}\}_{j=1}^m$, the whole training data set is described by~$\mathcal D=\{X,Y\}$ with the input training matrix~$X=[\bm x^{\{1\}},\bm x^{\{2\}},\ldots,\bm x^{\{m\}}]\in\R^{n\times m}$ and the output training matrix~$Y=[\tilde{\y}^{\{1\}},\tilde{\y}^{\{2\}},\ldots,\tilde{\y}^{\{m\}}]^\top\in\R^{m\times n}$. Now, the objective is to predict the output of the function~$\y^*$ at a test input~$\x^*\in\R^n$.\\
The underlying assumption of Gaussian process regression is that the data can be represented as a sample of a multivariate Gaussian distribution. The joint distribution of the~$i$-th component of~$\y^*$ is 
\begin{align}
	\begin{bmatrix} Y_{:,i} \\ y^*_i \end{bmatrix}\sim \mathcal{N} \left(\bm m(\x), \begin{bmatrix} K_{\varphi_i}(X,X) & \bm k_{\varphi_i}(\x^*,X)\\ \bm k_{\varphi_i}(\x^*,X)\tran & 				k_{\varphi_i}(\x^*,\x^*) \end{bmatrix}\right)\label{for:joint}
\end{align} 
with the covariance function~$k_{\varphi_i}(\x,\x^\prime)\colon\R^n\times\R^n\to\R$ as a measure of the correlation of two points~$(\bm x,\bm x^\prime)$. The function~$K_{\varphi_i}(X,X)\colon\R^{n\times m}\times \R^{n\times m}\to\R^{m\times m}$ is called the covariance or Gram matrix 
\begin{align}
	K_{j,l}= k_{\varphi_i}(X_{:, l},X_{:, j})
\end{align}
with~$j,l\in\lbrace 1,\ldots,m\rbrace$ where each element of the matrix represents the covariance between two elements of the training data~$X$. The vector-valued covariance function~$\bm k_{\varphi_i}(\x,X)\colon\R^n\times \R^{n\times m}\to\R^m$ calculates the covariance between the test input~$\x^*$ and the training data~$X$  
\begin{align}
	\bm k_{\varphi_i}(\x^*,X)\text{ with }k_{\varphi_i,j} = k_{\varphi_i}(\x^*,X_{:, j})
\end{align}
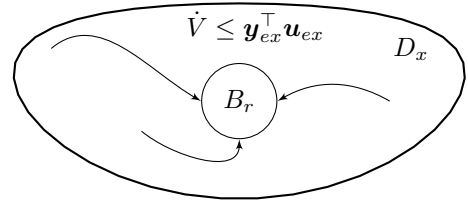
\begin{figure}[t]
\vspace{0.3cm}
	\begin{center}	
	 	\begin{tikzpicture}[auto,>=latex']
\draw [black, thick,  domain=0:360, samples=40] plot ({3*sin(\x)}, {1.3*cos(\x)+(3*sin(\x))^2/25} );
 \node[draw,circle,minimum height=1cm] (ball) {$B_r$};
\node[text width=0.5cm] at (2.3,0.7) {$D_x$};
\node[text width=2cm] at (0.3,0.99) {$\dot{V}\leq \y[ex]^\top\iu[ex]$};
\draw [->] (2,0) to [out=150,in=30] (ball.east);
\draw [->] (-2.5,0.7) to [out=40,in=150] (ball.west);
\draw [->] (-1.3,-0.4) to [out=-40,in=270] (ball.south);
\end{tikzpicture}
		\caption{Concept of semi-passivity. The system behaves passive in~$D_x\backslash B_r$.\label{fig:semipassiv}}
	\end{center}
	\vspace{-0.3cm}
\end{figure}
for all~$j\in\lbrace 1,\ldots,m\rbrace$ and~$i\in\lbrace 1,\ldots,n\rbrace$. These functions depend on a set of hyperparameters~$\varphi_i$ whose number of parameters depends on the function used. The choice of the covariance function and the corresponding hyperparameters can be seen as degrees of freedom of the regression. A comparison for the characteristics of the different covariance functions can be found in~\cite{bishop2006pattern}\\
The prediction of each component of~$\y^*$ is derived from the joint distribution~\cref{for:joint} and therefore it is a Gaussian distributed variable. The conditional probability distribution is defined by the mean
\begin{align}
	\mean(y^*_i\vert \x^*,\mathcal D)&=\bm k_{\varphi_i}(\x^*,X)\tran (K_{\varphi_i}+I \sigma^2_{i})^{-1}Y_{:,i},
\end{align}
where~$I$ is the identity matrix, and the variance
\begin{align}
	\var(y^*_i\vert \x^*,\mathcal D)&=k_{\varphi_i}(\x^*,\x^*)-\bm k_{\varphi_i}(\x^*,X)\tran \notag\\
	& \phantom{{}=}(K_{\varphi_i}+I \sigma^2_{i})^{-1} \bm k_{\varphi_i}(\x^*,X).
\end{align}
For the multi-variable Gaussian distribution, the~$n$ normally distributed components of~$\y^*\vert \x^*,\mathcal D$ are concatenated such that
\begin{align}
	\y^*\vert \x^*,\mathcal D &\sim \mathcal{N} (\bm\mean(\cdot),\Var(\cdot))\notag\\
	\bm \mean(\y^*\vert \x^*,\mathcal D)&=[\mean(y^*_1\vert \x^*,\mathcal D),\ldots,\mean(y^*_n\vert \x^*,\mathcal D)]\tran\notag\\
	\Var(\y^*\vert \x^*,\mathcal D)&=\diag(\var(y^*_1\vert \x^*,\mathcal D),\ldots,\var(y^*_n\vert \x^*,\mathcal D))\label{for:multigp},
\end{align} 
where the hyperparameters~$\varphi_i$ are optimized by means of the marginal likelihood function~\cite{rasmussen2006gaussian}. For this purpose, a gradient based algorithm is often used to find a (local) maximum of the marginal log-likelihood function
\begin{align}
	\varphi_i^* = \arg\max_{\varphi_i} \log P(Y_{:,i}|X,\varphi_i),\,\forall i=1,\ldots,n
\end{align}
to achieve suitable hyperparameters.
\section{Main Result}
\label{sec:modeling}
For the passivation of the system~\cref{for:system}, a closed loop with a GPR and a feedback control law is proposed. The GPR is used as feed-forward compensation of the unknown dynamics so that the drift function of the closed-loop is bounded. Based on this, a feedback control law is exploited to render the system strictly semi-passive. For this purpose, the input~$\iu$ of the system~\cref{for:system}
\begin{align}
	\iu=\iu[c]+\iu[gp]-\iu[ex],\label{for:inter}
\end{align}
is decomposed into a feed-forward component~$\iu[gp]\in\R^n$, a feedback control law~$\iu[c]\in\R^n$, and an additional external input~$\iu[ex]\in U\subset\R^n$, as shown in~\cref{fig:inter}. The control law is given by
\begin{align}
	\iu[c]=K_d\x_2+K_p\x_1\label{for:upd}
\end{align}
with positive definite, symmetric matrices~$K_p,K_d\in\R^{n\times n}$. 
\begin{rem}
For~\cref{for:system} with output~$\y=\x_1$, the control law~\cref{for:upd} is equivalent to a PD control law.
\end{rem}
For the rest of the paper, we assume the following properties for the passivation.
\begin{assum}\label{as:set}
	Consider the closed sets~$D_x\subseteq\X^2$ together with $D_{\dot{x}}\subset\R^{n}$ in the neighborhood of~$\bm 0$ such that
	\begin{align}
		\{\bm{z}\in\R^n\vert\Verts{z}<\bar{k}_p\Verts{\x_1}+\bar{k}_d\Verts{\x_2}+\Verts{\iu[ex]}\}\subseteq D_{\dot{x}},
	\end{align}
	holds for all $\x\in D_x,\iu[ex]\in U$ and the positive constants~$\bar{k}_d>c$ and~$\bar{k}_p>\max\{c\bar{k}_d^2/(4\bar{k}_d-4c),c^2\}$. 
\end{assum}
This assumption guarantees that~$\dx_2$ of the closed loop system is always element of~$D_{\dot{x}}$ which is required for the computation of the model error. The size of the set~$D_{\dot{x}}$ can be computed by~$D_x,U$ and~$c$.
\begin{assum}\label{as:mapping}
	For all~$\x\in D_x,\dx_2\in D_{\dot{x}}$ the mapping between~$\x,\dx_2$ and the input~$\iu$ must be unique, so that there exists a function~$\dyn^{-1}\colon D_x\times D_{\dot{x}}\to \R^n$ with
	\begin{align}
		(\x,\dx_2)\mapsto \dyn^{-1}(\x,\dx_2)=\iu
	\end{align}
\end{assum}
Thus,~\cref{for:system} is restricted to systems which are explicitly solvable to the input~$\iu=\dyn^{-1}(\x,\dx_2)$ on~$D_x \times D_{\dot{x}}$. This assumption holds for a large class of dynamical systems such as control affine systems with fully-ranked input matrix, e.g. many Lagrangian systems. With~\cref{as:set,as:mapping}, the system~\cref{for:system} with input~\cref{for:inter} can be rewritten as 
\begin{align}
\dx_1&=\x_2\notag\\
\dx_2&=\tilde{\dyn}(\x,\dx_2)-K_d\x_2-K_p\x_1-\iu[gp]+\iu[ex]\label{for:cl}
\end{align}
for all~$\x\in D_x$. The function~$\tilde{\dyn}\colon D_x\times D_{\dot{x}}\to \R^n$ is defined as~$\tilde{\dyn}(\x,\dx_2)=\dyn^{-1}(\x,\dx_2)+\dx_2$. The output of the GPR~$\iu[gp]$ is produced with the predicted mean of the unknown function values~$\y=\tilde\dyn(\x,\dx_2),\y\in\R^n$, i.e. 
\begin{align}
	\iu[gp]=\Mean(\y\vert\x,\dx_2,\mathcal D)\label{for:ugp}
\end{align}
computed with~\cref{for:multigp} and based on a the current state~$\x$ and~$\dx_2$ of the system. For this purpose, the training data~$\mathcal D$ set of the Gaussian process model is based on~$m$ training data pairs
\begin{align}
	\mathcal D=\left\lbrace \begin{bmatrix} \dx_2 \\ \x \end{bmatrix}^{\{i\}},\tilde\y^{\{i\}}\right\rbrace_{i=1}^{m}\label{GPlearning}
\end{align}
where the training data~$\tilde\y=\iu-\dx_2+\bm\epsilon$ of the system~\cref{for:system} is corrupted by Gaussian noise~$\bm\epsilon\sim\mathcal{N}(\bm 0,\diag (\sigma_{1}^2,\ldots,\sigma_{n}^2))$. The data can be generated by using any controller that behaves well-enough to produce a finite set of training points.
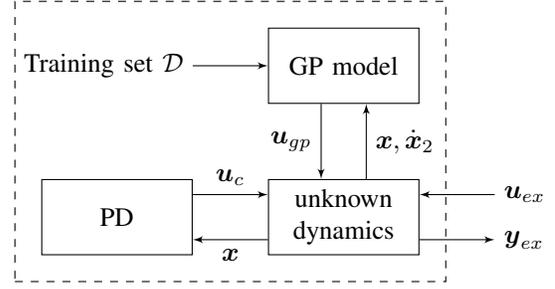
\begin{figure}[t]
\vspace{0.2cm}
	\begin{center}	
	 	\begin{tikzpicture}[auto,>=latex']
\tikzstyle{block} = [draw, fill=white, rectangle, minimum height=1cm, minimum width=2cm, align=center]
\tikzstyle{sum} = [draw, fill=white, circle, node distance=1cm]
\tikzstyle{input} = [coordinate]
\tikzstyle{output} = [coordinate]
\tikzstyle{t_output} = []
\tikzstyle{pinstyle} = [pin edge={to-,thin,black}]

    \node [block] (controller) {PD};
    \node [block, right= 1cm of controller] (system) {unknown\\dynamics};
    \node [block, above= 1 cm of system] (GPR) {GP model};
    
    \node [input, above right = 0.3cm and 1cm of system.east] (input) {};
    \node [output, below = 0.6cm of input] (output) {};
    \node [input, left= 1cm of GPR] (training) {};
    
    \draw [->] (training) -- node[anchor=east,pos=0,align=left] {Training set $\mathcal D$} (GPR.west);
	\draw [->] (input) -- node[anchor=west,pos=0] {$\iu[ex]$} ([yshift=0.3cm]system.east);
	\draw [<-] (output) -- node[anchor=west,pos=0] {$\y[ex]$} ([yshift=-0.3cm]system.east);

	\draw [->] ([yshift=0.3cm]controller.east) -- node[above] {$\iu[c]$} ([yshift=0.3cm]system.west);
	\draw [<-] ([yshift=-0.3cm]controller.east) -- node[below] {$\x$} ([yshift=-0.3cm]system.west);

	\draw [->] ([xshift=-0.3cm]GPR.south) -- node[left] {$\iu[gp]$} ([xshift=-0.3cm]system.north);
	\draw [<-] ([xshift=0.3cm]GPR.south) -- node[right] {$\x,\dx_2$} ([xshift=0.3cm]system.north);
		
	\node[draw,dashed,inner xsep=0.35cm,inner ysep=0.35cm,fit=(controller) (system) (GPR)] {};

\end{tikzpicture}
		\caption{Semi-passively rendered w.r.t.~$\iu[ex]$ and~$\y[ex]$.\label{fig:inter}}
	\end{center}
	\vspace{-0.7cm}
\end{figure}
\subsection{Model error}
After the learning procedure, it is possible to compute an upper bound for the error between the mean prediction~$\iu[gp]=\Mean(\y\vert\x,\dx_2,\mathcal D)$ of the Gaussian process model and the function $\tilde{\dyn}$ of~\cref{for:cl}. For this purpose, the covariance function must be selected in such a way that the function~$\tilde{\dyn}$ is an element of the associated RKHS.
\begin{assum}\label{as:rkhs}
The function~$\tilde{\dyn}(\x,\x_2)$ has a bounded reproducing kernel Hilbert space (RKHS) norm in respect to the covariance function~$k(\cdot,\cdot)$, so that~$\Verts{\tilde{\dyn}}_k<\infty$ on~$D_x\times D_{\dot{x}}$.
\end{assum}
This seems to be paradoxical since the function is assumed to be unknown. However, there exist some covariance functions, so called universal kernel functions, which can approximate any continuous function arbitrary precisely on a compact set~\cite[Lemma 4.55]{steinwart2008support}, e.g. the squared exponential covariance function. Therefore, many dynamics can be covered by the universal covariance function so that this assumption is not at all restrictive. A more detailed discussion about RKHS norms and covariance functions is given by~\cite{wahba1990spline}\\
 A variance dependent bound for the scalar case is presented in~\cite{srinivas2012information} and is here extended to an absolute bound for multidimensional predictions in the following lemma.
\begin{lem}
	\label{lemma:boundederror}
	Consider the system~\cref{for:system} satisfying~\cref{as:rkhs,as:set,as:mapping} and a Gaussian process model based on~\cref{GPlearning}. The model error is bounded with a~$\bm\Delta\in\R_{\geq 0}^n,\bar\Delta\in\R_{\geq 0}$ by\footnote{For notational reasons, we suppress the conditional part of the predicted mean and variance}
	\begin{align}
	\text{P}\left\lbrace \Verts{\Mean(\y)-\tilde\dyn(\x,\dx_2)}\leq \Verts{\bm\Delta^\top\Var^{\frac{1}{2}}(\y)}  \right\rbrace&\geq \delta\notag\\
		 \Verts{\bm\Delta^\top\Var^{\frac{1}{2}}(\y)}&\leq\bar\Delta\label{for:upbound}
	\end{align}
	for all~$\x\in D_x,\dx_2\in D_{\dot{x}}$ with~$\delta\in(0,1)$.
\end{lem}
\begin{proof}
Following \cite[Theorem 6]{srinivas2012information}, the elements of~$\Delta$ are defined by
	\begin{align}
		\Delta_j&=\sqrt{2\Verts{\tilde{f}_j}^2_k+300 \gamma_j \ln^3\left(\frac{m+1}{\delta}\right)}
		\end{align}
		where~$\gamma_j\in\R$ is the maximum information gain, i.e.
		\begin{align}
		\gamma_j&=\max_{X}\frac{1}{2}\log \vert I+\sigma_i^{-2}K_{\varphi_j}(X,X)\vert\\
		X&=\begin{bmatrix}
		\begin{bmatrix}\dx_2\\ \x\end{bmatrix}^{\{1\}},\ldots,\begin{bmatrix}\dx_2\\ \x\end{bmatrix}^{\{m+1\}}
		\end{bmatrix}\text{, where}\\
		&\begin{bmatrix}\dx_2\\ \x\end{bmatrix}^{\{i\}}\in D_x\times D_{\dot{x}}
	\end{align}
for~$i=1,\ldots,m+1$. With~\cref{as:rkhs,as:set} and the fact that~$\bm\epsilon$ is uncorrelated, the model error of a multidimensional prediction for all~$(\x,\dx_2)\in D_x\times D_{\dot{x}}$ is given by
\begin{align}
\text{P}\left\lbrace \smashoperator{\bigcap\limits_{\qquad j=1,\ldots,n}} \vert\mean(y_j)-\tilde{f}_j(\x,\dx_2)\vert\leq \vert \Delta_j \var^{\frac{1}{2}}(y_j)\vert\right\rbrace&\geq (1-\delta_{sc})^n\notag\\
\Rightarrow\text{P}\left\lbrace \Verts{\Mean(\y)-\dyn(\x,\dx_2)}\leq \Verts{\bm \Delta\tran \Var^{\frac{1}{2}}(\y)}\right\rbrace&\geq (1-\delta_{sc})^n\label{for:prob}
\end{align}
with~$\delta_{sc}\in(0,1)$. Since~$\bm\Delta$ is finite and the variance is also bounded on a closed set~\cite{beckers:cdc2016}, it exists a constant~$\bar\Delta\in\R_{>0}$ which bounds~$\Verts{\bm \Delta\tran \Var^{1/2}(\y)}\leq\bar\Delta$ in~\cref{for:prob}. Thus, the model error is bounded with a probability of at least~$\delta\coloneqq (1-\delta_{sc})^n$ by~$\bar\Delta$.
\end{proof}
\begin{rem}
The information capacity~$\bm\gamma$ has a sub-linear dependency on the number of training points for many commonly used covariance functions, e.g. the squared exponential covariance function, and can be bounded by a constant~\cite{srinivas2012information}. Therefore, even though~$\Verts{\bm \Delta}$ is increasing with the number of training data, it is possible to learn the true function~$\tilde{\dyn}(\x,\dx_2)$ arbitrarily exactly~\cite{Berkenkamp2016ROA}. 
\end{rem}
The result of~\cref{lemma:boundederror} is a an upper bound for the model error. The stochastic nature of the bound is due to the fact that just a finite number of noisy training points are available and thus, the true function cannot be known exactly. If exact knowledge of the model was available, the variance of the GPR would be zero and thus, the upper bound for the model error would also be zero. With an increasing number of training points or decreasing noise~$\sigma$ of the training data, the bound becomes tighter~\cite{umlauft:cdc2017}. Since the model is used for a feed-forward compensation of the unknown dynamics of the systems, the model error directly effects the size of the set where the system behaves passive as shown in the next section.
\subsection{Passivation}
\label{sec:passive}
Before we present the main theorem about the passivity of the closed loop system, the following definition and lemmas are introduced.
\begin{defn}
Let~$\Lambda$ be a matrix-valued function which maps from $\R^{n\times n}\times\R^{n\times n}\times\R_{>0}\to\R^{2n\times 2n}$ with
\begin{align}
\Lambda(K_d,K_p,c)\coloneqq \begin{bmatrix} K_d-c I & \frac{c}{2}K_d\\\frac{c}{2}K_d & c K_p\end{bmatrix}.
\end{align}
\end{defn}
\begin{lem}
\label{lem:pdf}
For any~$c,\leig_d\in\R_{>0}$, there exist positive definite and symmetric matrices~$K_d,K_p\in\R^{n\times n}$, so that 
\begin{align}
\leig\left(\Lambda(K_d,K_p,c)\right)\geq \leig_d.
\end{align}
\end{lem}
\begin{proof}
Assuming the positive definite, symmetric matrices~$\tilde{K}_d,\tilde{K}_p\in\R^{n\times n}$. The matrix~$\tilde{\Lambda}_M\in\R^{2n\times 2n}$ with
\begin{align}
\tilde{\Lambda}_M=\begin{bmatrix} \tilde{K}_d & \frac{c}{2}\tilde{K}_d\\\frac{c}{2}\tilde{K}_d & c \tilde{K}_p\end{bmatrix}
\end{align}
is positive definite, if~$c\tilde{K}_p\succ 0$ and 
\begin{align}
\underbrace{\tilde{K}_d-c\frac{\tilde{K}_d \tilde{K}_p^{-1} \tilde{K}_d}{4}}_{\tilde{\Lambda}_S\in\R^{n\times n}}\succ 0
\end{align}
using the property of the Schur complement. The eigenvalues of~$\tilde{\Lambda}_S$ are lower bounded by
\begin{align}
\lambda_i(\tilde{\Lambda}_S)\geq \leig(\tilde{K}_d)-c\frac{\geig^2(\tilde{K}_d)}{4\leig(\tilde{K}_p)}
\end{align}
and thus, it is always possible to select a~$\tilde{K}_p$, such that the matrix~$\tilde{\Lambda}_S\succ 0$ and, consequently,~$\tilde{\Lambda}_M\succ 0$. Now, assume a scaling factor~$\gamma\in\R_{\geq 0}$. The eigenvalues of the overall sum~$\tilde{\Lambda}\in\R^{n\times n}$ of the two symmetric matrices
\begin{align}
\tilde{\Lambda}=\gamma\begin{bmatrix} \tilde{K}_d & \frac{c}{2}\tilde{K}_d\\\frac{c}{2}\tilde{K}_d & c \tilde{K}_p\end{bmatrix}+\begin{bmatrix} -c I & 0\\ 0 & 0\end{bmatrix}
\end{align}
are lower bounded by
\begin{align}
\leig(\tilde{\Lambda})\geq -c+\gamma\leig(\tilde{\Lambda}_M).
\end{align}
Since~$\tilde{\Lambda}_M\succ 0$, for any $c,\leig_d$ there exist a~$\gamma$ such that the eigenvalue~$\leig(\tilde{\Lambda})\geq\leig_d$. Finally, defining~$K_d=\gamma \tilde{K}_d$ and~$K_p=\gamma \tilde{K}_p$ concludes the proof.
\end{proof}
\begin{lem}
\label{lem:pdf2}
For all~$c\in\R_{>0}$, there exist positive definite and symmetric matrices~$K_d,K_p\in\R^{n\times n}$ with
\begin{align}
\geig(K_d)&\leq \bar{k}_d\in\R_{>0},\,\bar{k}_d>c\\
\geig(K_p)&\leq \bar{k}_p\in\R_{>0},\,\bar{k}_p>\frac{c}{4}\frac{\bar{k}_d^2}{\bar{k}_d-c},
\end{align}
such that~$\Lambda(K_d,K_p,c)\succ 0$.
\end{lem}
\begin{proof}
The matrix~$\Lambda(K_d,K_p,c)$ is positive definite, iff~$cK_p\succ 0$ that is fulfilled by definition, and
\begin{align}
\underbrace{\tilde{K}_d-cI-c\frac{K_d K_p^{-1} K_d}{4}}_{\Lambda_S\in\R^{n\times n}}\succ 0 .
\end{align}
Analogous to the proof of~\cref{lem:pdf}, the eigenvalues of~$\Lambda_S$ are lower bounded by
\begin{align}
\lambda_i(\Lambda_S)\geq \leig(K_d-cI)-c\frac{\geig^2(K_d)}{4\leig(K_p)},
\end{align}
so that is is possible to achieve~$\leig(\Lambda_S)>0$ with matrices~$K_d,K_p$ which satisfy~$\geig(K_d)\leq \bar{k}_d$ and~$\geig(K_p)\leq \bar{k}_p$.
\end{proof}
\begin{thm}
\label{thm:main}
Given~\cref{as:rkhs,as:set,as:mapping} and the closed loop system~\cref{for:inter}. Then, there exist positive definite, symmetric matrices~$K_p, K_d$ and a maximal model error~$\bar\Delta$, so that~\cref{for:system} is rendered strictly semi-passive with 
\begin{align}
B_r&=\sqrt{\frac{(1+c)\bar\Delta}{\leig\left(\Lambda(K_d,K_p,c)\right)}}
\end{align}
on the set~$D_x$ with a given probability~$\delta\in(0,1)$.
\end{thm}
\begin{proof}
We assume the storage function 
\begin{align}
V(\x)=\frac{1}{2}\x_1^\top K_p\x_1+\frac{1}{2}\x_2^\top \x_2+c \x_2^\top \x_1,\label{for:V}
\end{align}
that is positive for~$\leig(K_p)> c^2$ for all~$\x_2,\x_1\in\R^n$ and zero for~$\x_2=\x_1=\bm{0}$. With~\cref{for:cl} the derivative of~$V$ is given by
\begin{align}
\dot V(\x)&=-\begin{bmatrix}\x_2^\top & \x_1^\top\end{bmatrix}
\Lambda(K_d,K_p,c)
\begin{bmatrix}\x_2 \\ \x_1\end{bmatrix}\notag\\
&+(\x_2+c\x_1)^\top(\tilde{\dyn}(\x,\dx_2)-\Mean(\y)+\iu[ex]).\label{for:Vdot}
\end{align}
The first term of the equation depends on the feedback gains whereas the second term depends on the model error. Following~\cref{lem:pdf}, for any~$c$ there exist two matrices~$K_d$ and~$K_p$, so that the matrix~$\Lambda$ is positive definite. The error between the true dynamics and the mean of the GPR in~\cref{for:Vdot} is bounded by a constant~$\bar\Delta\in\R_{>0}$ with the probability~$\delta$ using~\cref{lemma:boundederror}. Thus, the drift of the Lyapunov function is bounded by a function~$h\colon D_x\to\R$ with
\begin{align}
\dot V(\x_2,\x_1)&\leq \y[ex]^\top\iu[ex]- h(\x_2,\x_1)\\
h(\x_2,\x_1)&=\leig(\Lambda)\Verts{\begin{matrix}\x_2\\ \x_1\end{matrix}}^2-\bar\Delta\Verts{\x_2}-c\bar\Delta\Verts{\x_1}.
\end{align}
The function~$h$ is positive for
\begin{align}
\Verts{\begin{matrix}\x_2\\ \x_1\end{matrix}}> \sqrt{\frac{(1+c)\bar\Delta}{\leig(\Lambda)}}=r,\label{for:radius}
\end{align}
i.e. outside a ball~$B_r$ with the radius~$r\in\R_{>0}$. Finally, it must be guaranteed that a) the state~$\x$, once in~$D_x$, remains inside~$D_x$ while b)~$\dx_2\in D_{\dot x}$, so that the conditions of~\cref{lemma:boundederror} are not violated. The inequality~\cref{for:radius} shows that for any positive definite matrix~$\Lambda(K_d,K_p,c)$, it is possible to find a~$\bar\Delta$ so that~$r$ is arbitrary small. As consequence, there exists a~$\bar\Delta$, so that the ball~$B_r$ is a subset of~$D_x$ and thus, the state~$\x$ remains in~$\in D_x$.\\
To guarantee that~$\dx_2\in D_{\dot x}$, we use the closed loop dynamics~\cref{for:cl} with the maximum model error~$\bar\Delta$ to compute an upper bound for~$\Verts{\dx_2}$ which is given by
\begin{align}
\Verts{\dx_2}&\leq\Verts{\bar\Delta-K_d\x_2-K_p\x_1+\iu[ex]}\\
&\leq\Verts{\bar\Delta}+\geig(K_d)\Verts{\x_2}+\geig(K_p)\Verts{\x_1}+\Verts{\iu[ex]}.
\end{align}
With~\cref{lem:pdf2,as:set}, there exist a~$K_p,K_d$ so that~$\Lambda(K_d,K_p,c)$ is positive definite and~$\dx_2\in D_{\dot x}$ for all~$\x\in D_x$. Therefore, the system~\cref{for:system} is rendered strictly semi-passive with the probability~$\delta$ in respect to~$\iu[ex]$ and~$\y[ex]$.
\end{proof}
\begin{rem}
The radius of the Ball~$B_r$ can be set arbitrary small by either decreasing the maximum model error~$\bar\Delta$ or increasing the feedback gains~$\leig(\Lambda)$.
\end{rem}

\section*{Simulation}
For the simulation, we use a modified Duffing oscillator 
\begin{align}
\dot{x_1}&=x_2\\
\dot{x_2}&=u^{1/3}-\gamma x_2-\alpha x_1-\beta x_1^3+1
\end{align}
as sample system where not only the parameters are unknown but also the entire parametric form of the dynamics is assumed to be unknown. This nonlinear, second-order system describes the motion of a damped oscillator with a more complex potential than in simple harmonic motion. The parameters are set to~$\alpha=-0.1,\beta=-0.1$,~$\gamma=0.1$, such that the system's equilibrium point is unstable, see~\cref{fig:stream0}. The control input is chosen to be not input affine to demonstrate the efficiency of the proposed method. Now, the passivation approach of~\cref{thm:main} is applied. We set~$c=0.5$ for the passive output and~$\max(\vert u_\text{ex}\vert)=0.1$ for the passive input. Additionally, we set
\begin{align}
\bar{k}_d&\coloneqq 0.9>0.5=c\\
\bar{k}_p&\coloneqq 0.254>0.253=\max\{c\bar{k}_d^2/4(\bar{k}_d-c),c^2\},
\end{align}
so that~$D_{\dot{x}}=[-2.55,4.55]$ fulfills~\cref{as:set}. Since the drift function of the oscillator is continuous, the squared exponential covariance function for the Gaussian process regression is used to learn~$\tilde{f}$. For this purpose, we generate 720 pairs of inputs~$\{\dot{x}_2,x_1,x_2\}$ and outputs~$\{u-\dot{x}_2\}$ as training data on~$\x\in [-2,2]^2=D_x$ and~$\dot{x}_2\in D_{\dot{x}}$. The hyperparameters of the squared exponential covariance function are optimized by a descent gradient algorithm.\\
The feedback gains are set to~$K_d=0.9$ and~$K_p=1$. In combination with the maximum model error~$\bar\Delta=0.045$ on the set~$D_x\times D_{\dot{x}}$, the state's derivative~$\dot{x}_2$ of the passive  system is element of~$D_{\dot{x}}$, see~\cref{fig:surf}. In addition, the ball~$B_r$ is a subset of~$D_x$ which is visualize in~\cref{fig:stream1} together with the phase plane of the Duffing oscillator that is rendered strictly semi-passive. The result is that inside the set~$D_x\backslash B_r$, the closed loop system is behaves passive.
\begin{figure}[t]
\vspace{0.1cm}
	\begin{center}	
	 	\begin{tikzpicture}
\begin{axis}[
	height=5.5cm,
	width=\columnwidth,
  xmin=-4,xmax=4,
  ymin=-4,ymax=4,
  xlabel={$x_1$},
  ylabel={$x_2$},
]
  \foreach \num in {1,2,...,38} {
    \addplot[blue] file {./data/vertices0/vertice_\num.dat};                          
  }
   \foreach \numx in {0,2,...,96} {
   \pgfmathsetmacro{\numy}{int(\numx+1)}
    \addplot[blue] table [x index=\numx,y index=\numy]{./data/vertices0/arrows.dat};                          
  }
\end{axis}
\end{tikzpicture}
		\vspace{-0.3cm}\caption{Phase plane portrait of the uncontrolled Duffing oscillator.\label{fig:stream0}}
	\end{center}
	\vspace{-0.3cm}
\end{figure}
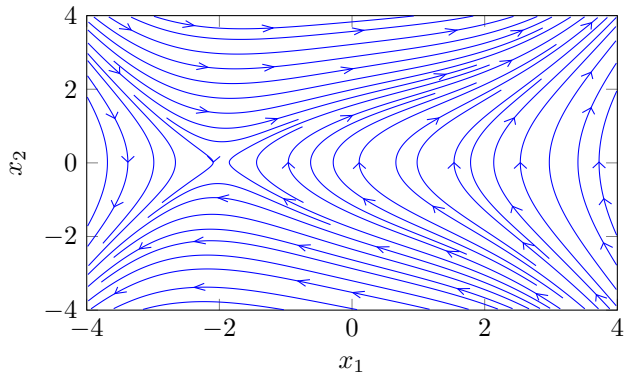
\begin{figure}[t]
	\begin{center}	
	 	\begin{tikzpicture}
    \begin{axis}[width=0.8\columnwidth,
    				height=4cm,
    				view={-150}{20}, 
    				colormap/cool, 
    				xlabel={$x_1$},
    				ylabel={$x_2$},
    				zlabel={$\dot{x}_2$},
    				xlabel shift=-5pt,
    				ylabel shift=-5pt,
    				zlabel shift=-5pt,
    				scale only axis]
\addplot3[surf,mesh/cols=30] table {data/surf.dat};
\addplot3[red,line width=1.5] coordinates {
        (-2,  2,  -0.022)
        (-2,  2,  4.55)
        (2,  2,  4.55)
        (2,  2,  -1.933)
	};
\addplot3[red,line width=1.5] coordinates {
        (-2,  2,  -0.022)
        (-2,  2,  4.55)
        (-2,  -2,  4.55)
        (-2,  -2,  3.381)
	};
\addplot3[red,line width=1.5] coordinates {
        (2,  -2,  0.096)
        (2,  -2,  4.55)
        (-2,  -2,  4.55)
        (-2,  -2,  3.381)
	};
\addplot3[red,line width=1.5] coordinates {
        (2,  -2,  0.096)
        (2,  -2,  4.55)
        (2,  2,  4.55)
        (2,  2,  -1.933)
	};
\node[text=red,align=center] at (2,-4,5) {$D_x\times D_{\dot{x}}$};
\end{axis}
\end{tikzpicture} 
		\vspace{-0.3cm}\caption{The figure demonstrates that with the selected~$K_p$,~$K_d$ and~$\bar\Delta$, the state's derivative~$\dot{x}_2$ is element of~$D_{\dot{x}}$ on~$D_x$\label{fig:surf}}
	\end{center}
	\vspace{-0.4cm}
\end{figure}
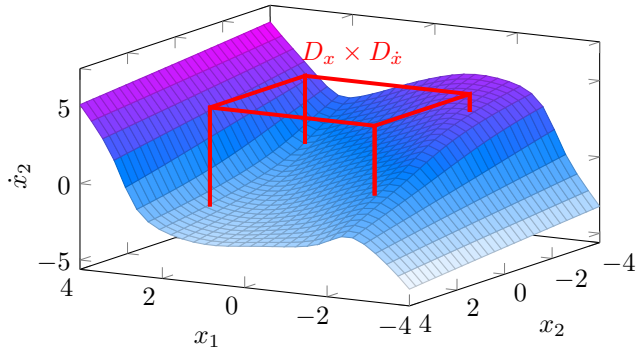
\begin{figure}[t]
\vspace{0.1cm}
	\begin{center}	
	 	\begin{tikzpicture}
\begin{axis}[
height=5.5cm,
width=\columnwidth,
  xmin=-4,xmax=4,
  ymin=-4,ymax=4,
  xlabel={$x_1$},
  ylabel={$x_2$},
]
  \foreach \num in {1,2,...,45} {
    \addplot[blue] file {./data/vertices1/vertice_\num.dat};                          
  }
     \foreach \numx in {0,2,...,96} {
   \pgfmathsetmacro{\numy}{int(\numx+1)}
    \addplot[blue] table [x index=\numx,y index=\numy]{./data/vertices1/arrows.dat};                          
  }
\node[fill=white,text=red,align=left] at (0.8,0) {$B_r$};
\node[fill=white,text=red,align=left] at (1.55,1.6) {$D_x$};
\draw [red,line width=1.5] (-2,-2) rectangle (2,2); 
\draw [red,line width=1.5] (0,0) circle (0.55); 
\end{axis}
\end{tikzpicture}
		\vspace{-0.3cm}\caption{The closed loop system with the Duffing oscillator is strictly semi-passive in~$D_x$ with the ball~$B_r$. \label{fig:stream1}}
	\end{center}
	\vspace{-0.8cm}
\end{figure}
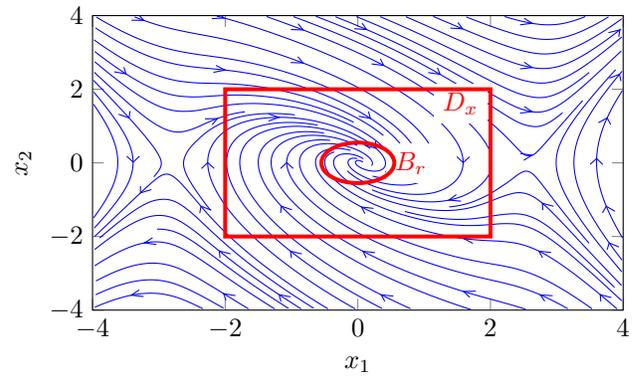
\label{sec:sim}

\section*{Conclusion}
In this paper, we present a data-driven method to render a class of nonlinear systems with unknown dynamics strictly semi-passive. As consequence, the closed-loop system behaves passive outside a ball~$B_r$ on a set~$D_x$. For this purpose, we use Gaussian process regression for the feed-forward compensation of the unknown dynamics and a feedback control law to render the closed loop system semi-passive. It is shown, that the radius of the ball~$B_r$ can be set arbitrary small depending on the model error and the feedback gains. Finally, a simulation demonstrates the presented theory.

\section*{ACKNOWLEDGMENTS}
The research leading to these results has received funding from the European Research Council under the European Union Seventh Framework Program (FP7/2007-2013) / ERC Starting Grant ``Control based on Human Models (con-humo)'' agreement n\textsuperscript{o}337654.


\bibliography{mybib}
\bibliographystyle{ieeetr}

\end{document}